\documentclass[11pt]{article}
\usepackage[margin=1in]{geometry}
\usepackage{amsmath,amssymb,amsthm,mathtools}
\usepackage{booktabs}
\usepackage{graphicx}
\usepackage{microtype}
\usepackage{enumitem}
\usepackage{url}
\usepackage[hidelinks]{hyperref}
\graphicspath{{figures/}}

\newtheorem{theorem}{Theorem}[section]

\newtheorem{proposition}[theorem]{Proposition}
\newtheorem{corollary}[theorem]{Corollary}
\newtheorem{definition}[theorem]{Definition}

\newcommand{\MMS}{\operatorname{MMS}}
\newcommand{\RMMS}{\operatorname{RMMS}}
\newcommand{\EFL}{\operatorname{EFL}}
\newcommand{\EFX}{\operatorname{EFX}}
\newcommand{\bp}{\operatorname{bp}}
\newcommand{\cov}{\operatorname{cov}}

\title{Residual Maximin Share: Exact Finite-Agent Frontier, Sparse Extremizers, and Threshold Cuts}
\author{%
  Qinghua Qin$^{1,2,3}$\thanks{Email: \texttt{tsinghuaqin@pku.org.cn}, \texttt{qinqinghua@mail.las.ac.cn}.}\\[4pt]
  \small $^{1}$University of Chinese Academy of Sciences \quad $^{2}$Chinese Academy of Sciences \quad $^{3}$Huawei Technologies Co., Ltd.
}
\date{August 2026}

\begin{document}
\maketitle

\begin{abstract}
Residual maximin share ($\RMMS$) is the largest share threshold that remains guaranteeable throughout dynamic allocation processes, even after previously allocated, lower-valued bundles are removed from the item pool. For additive valuations, recent density-balance analyses established finite-agent lower bounds comparing $\RMMS$ with the classical maximin share ($\MMS$). In this paper, we prove that these finite-agent lower bounds are exact. Specifically, if $d_n$ denotes the largest odd integer at most $n$, the worst-case ratio satisfies
\[
\inf_{M,v:\MMS>0}\frac{\RMMS(M,v,n)}{\MMS(M,v,n)} = \lambda_n := \frac{2d_n}{3d_n-1}.
\]
Consequently, the exact additive frontier forms consecutive odd-even plateaus and converges monotonically to $2/3$.

We then investigate the combinatorial structure of extremal instances. While naive witnesses require $\Theta(n^2)$ items, we construct an explicit three-valued family achieving the exact boundary with only linear support: $(5n-3)/2$ items for odd $n$ and $(5n-4)/2$ items for even $n$. Its low-valued block supports two exact partitions that simultaneously certify the $\MMS$ benchmark and the residual obstruction. By modeling these dual partitions as a bipartite transportation graph, we prove that this block attains the absolute minimum support $q+d-1=3q$. At minimum support, any two-valued filler is uniquely rigid up to relabeling.

Finally, we establish structural characterizations of $\RMMS$. A general min--max representation applies to all finite monotone valuations. For integer additive valuations, we prove that a threshold $T$ is residual self-feasible if and only if every subset cut satisfies a packing-covering condition. Because $\RMMS$ is pointwise maximal among residual self-feasible shares, these exact constants establish a tight limitation on the fairness guarantees achievable by share-based lone-divider algorithms.
\end{abstract}

\section{Introduction}

The fair allocation of indivisible items is a central topic in computer science and economic theory, with practical applications in sequential resource allocation, cloud computing cluster assignment, inheritance division, and course registration systems~\cite{Budish2011,KurokawaEtAl2018,AmanatidisSurvey2023}. In these settings, items such as computing servers, course seats, or physical assets cannot be divided fractionally. A foundational static fairness benchmark is the \emph{maximin share} ($\MMS$), introduced by Budish~\cite{Budish2011}, which measures the value an agent can guarantee by partitioning all goods into $n$ bundles and receiving the least favorable one.

While static $\MMS$ provides a natural benchmark for one-shot allocations, dynamic and recursive allocation algorithms—such as the classical \emph{lone-divider} method and sequential picking protocols—face a structural challenge: when earlier agents select their preferred bundles and leave, the remaining item pool is diminished. If earlier agents remove bundles that a remaining agent values below the static $\MMS$ threshold, the residual goods may no longer allow the remaining $n-k$ agents to each receive a bundle of value at least $\MMS$.

\paragraph{Dynamic fairness and Residual Maximin Share.}
To address dynamic and recursive allocation settings, Feige~\cite{Feige2025} and Akrami, Feige, Mahara, Mehlhorn, and Rathi~\cite{AkramiEtAl2026} introduced the \emph{Residual Maximin Share} ($\RMMS$). A share threshold is \emph{residual self-feasible} if, whenever any $k$ bundles worth strictly less than the threshold are removed, the remaining goods can still be partitioned into $n-k$ bundles each worth at least the threshold. $\RMMS$ is defined as the supremum over all residual self-feasible thresholds.

Akrami et al.~\cite{AkramiEtAl2026} established that for additive valuations, $\RMMS$ satisfies the finite-agent lower bound:
\begin{equation}\label{eq:known-lower}
\RMMS(M,v,n)\ge
\begin{cases}
\dfrac{2n}{3n-1}\MMS(M,v,n),&n\text{ odd},\\[8pt]
\dfrac{2n-2}{3n-4}\MMS(M,v,n),&n\text{ even}.
\end{cases}
\end{equation}
Whether these finite-agent lower bounds are tight remained an open question.

\paragraph{Exact quantitative frontier.}
In this work, we prove that the lower bound~\eqref{eq:known-lower} is exact for every $n\ge2$. Let $d_n$ denote the largest odd integer at most $n$. We show that:
\begin{equation}\label{eq:exact-frontier-intro}
\boxed{
\inf_{M,v:\MMS>0}
\frac{\RMMS(M,v,n)}{\MMS(M,v,n)}
=\lambda_n:=\frac{2d_n}{3d_n-1}.
}
\end{equation}
The exact frontier exhibits an \emph{odd-even plateau} structure: consecutive odd and even values share the identical ratio,
\[
\lambda_{2r+1}=\lambda_{2r+2}=\frac{2r+1}{3r+1},
\]
which decreases monotonically to $2/3$. For example, for $n=3,4$ the exact ratio is $3/4 = 0.75$; for $n=5,6$ it is $5/7 \approx 0.7143$; for $n=7,8$ it is $7/10 = 0.70$.

\paragraph{Extremal instances and linear sparsity.}
A basic bi-valued construction matching $\lambda_n$ can be built using $2p-1$ high-value goods and $qd$ unit goods (where $q=\lfloor(n-1)/2\rfloor$, $d=2q+1$, and $p=n-q$), which requires $\Theta(n^2)$ items. We show that this quadratic number of goods is unnecessary by constructing a sparse, three-valued instance with \emph{linear support}:
\[
|M|=
\begin{cases}
(5n-3)/2,&n\text{ odd},\\
(5n-4)/2,&n\text{ even}.
\end{cases}
\]
The sparse instance uses $2p-1$ high goods of value $d$, $2q$ goods of value $q$, and $q$ unit goods. Its low-valued block admits two simultaneous exact partitions:
\begin{enumerate}[label=(\roman*),leftmargin=2em]
\item For the $\MMS$ certificate, it splits into $d$ bundles of value $q$, which pair with high goods to form $n$ bundles of value $d+q$;
\item For the $\RMMS$ obstruction, it splits into $q$ bundles of value $d=2q+1$. Removing these $q$ bundles leaves $2p-1$ high goods for $p$ remaining agents. Any threshold above $d$ would require two high goods per bundle (needing $2p$ high goods), which fails by exactly one good.
\end{enumerate}

\paragraph{Transportation graph minimality and rigidity.}
By mapping items to edges in a weighted bipartite transportation graph between the $d$ column bundles (value $q$) and $q$ row bundles (value $d$), coprimality $\gcd(q,d)=1$ forces the graph to be connected. Hence any valid filler must contain at least $q+d-1=3q$ items. Our construction meets this bound with equality, so its transportation graph is a spanning tree. Furthermore, under the two-value restriction, the $(q,q,1)$ configuration is uniquely forced up to relabeling.

\paragraph{Structural representations and threshold cuts.}
We also establish two structural characterizations of $\RMMS$:
\begin{enumerate}[label=(\alph*),leftmargin=2em]
\item \textbf{Finite Min--Max Representation:} For every finite monotone valuation, $\RMMS$ equals the minimum over all removal histories of the maximum between the costliest removed bundle and the residual $\MMS$.
\item \textbf{Threshold Packing-Covering Cut:} For integer additive valuations, a threshold $T$ is residual self-feasible if and only if for every subset $U\subseteq M$,
\[
\bp_T(U)+\cov_T(M\setminus U)\ge n,
\]
where $\bp_T(U)$ is the bin-packing number of $U$ with capacity $T-1$, and $\cov_T(M\setminus U)$ is the maximum number of disjoint value-$\ge T$ bundles in $M\setminus U$.
\end{enumerate}

\paragraph{Implications for the lone-divider framework.}
Because $\RMMS$ is pointwise maximal among all residual self-feasible shares and is self-maximizing~\cite{AkramiEtAl2026}, $\lambda_n$ is the exact upper bound on the $\MMS$ guarantee achievable by any dynamic algorithm whose invariant relies on residual self-feasibility.

\begin{figure}[t]
\centering
\includegraphics[width=.78\textwidth]{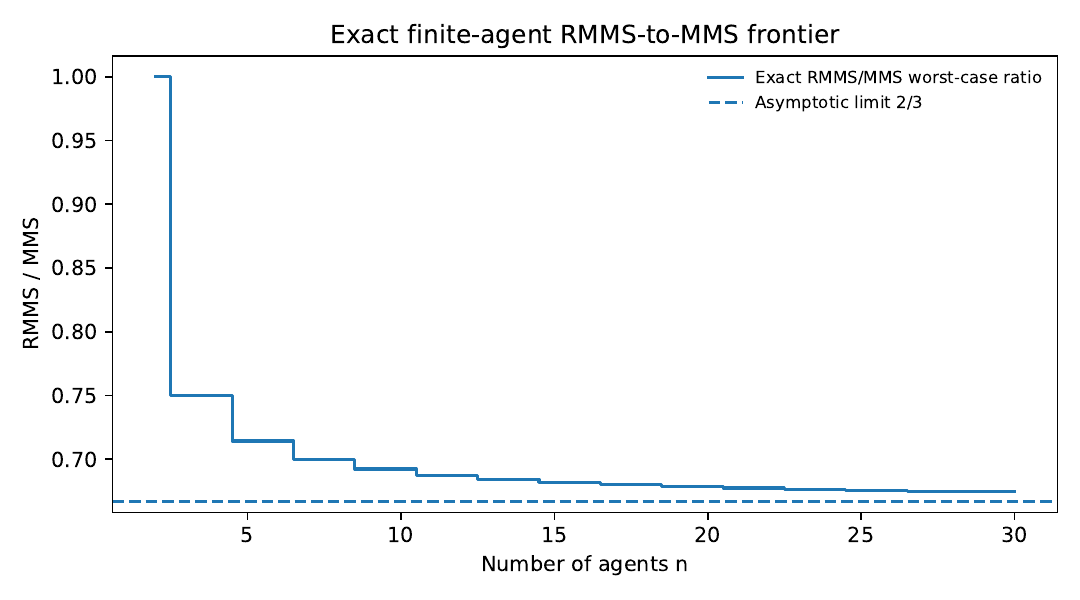}
\caption{The exact worst-case $\RMMS/\MMS$ frontier $\lambda_n$. Consecutive odd and even agent counts form plateaus, converging monotonically to $2/3$.}
\label{fig:frontier}
\end{figure}

\section{Preliminaries}

Let $M$ be a finite set of indivisible goods and let $v:2^M\to\mathbb R_{\ge0}$ be a normalized monotone valuation. An additive valuation satisfies $v(S)=\sum_{g\in S}v(g)$ for all $S\subseteq M$.

For any positive integer $r$ and $S\subseteq M$, the $r$-way maximin share is
\[
\MMS_r(S;v) = \max_{(P_1,\ldots,P_r)\in\Pi_r(S)} \min_{j\in[r]}v(P_j),
\]
and we write $\MMS(M,v,n)=\MMS_n(M;v)$.

\begin{definition}[Residual self-feasibility]
A share function $s(M,v,n)$ is \emph{residual self-feasible} if for every instance $(M,v,n)$, every integer $0\le k<n$, and every collection of $k$ pairwise disjoint bundles $B_1,\ldots,B_k$ satisfying $v(B_r)<s(M,v,n)$ for all $r\in[k]$, the residual item set $M\setminus\bigcup_{r=1}^k B_r$ can be partitioned into $n-k$ bundles each of value at least $s(M,v,n)$.
\end{definition}

\begin{definition}[Residual maximin share]
The \emph{residual maximin share} $\RMMS(M,v,n)$ is the supremum of all residual self-feasible thresholds for $(M,v,n)$.
\end{definition}

By definition, $\RMMS$ pointwise dominates every other residual self-feasible share. Feige~\cite{Feige2025} and Akrami et al.~\cite{AkramiEtAl2026} showed that $\RMMS$ is feasible and self-maximizing on all monotone valuations.

\section{A Finite Min--Max Representation}\label{sec:minmax}

For any collection of $k$ pairwise disjoint bundles $\mathcal B=(B_1,\ldots,B_k)$, define $R(\mathcal B)=M\setminus\bigcup_{r=1}^k B_r$ and
\begin{equation}\label{eq:phi-def}
\Phi(\mathcal B) = \max\left\{ \max_{r\in[k]}v(B_r),\ \MMS_{n-k}(R(\mathcal B);v) \right\},
\end{equation}
with $\max\varnothing=0$ when $k=0$.

\begin{theorem}[Finite min--max formula]\label{thm:minmax}
For every finite monotone valuation $v$,
\begin{equation}\label{eq:minmax}
\boxed{
\RMMS(M,v,n) = \min_{0\le k<n}\ \min_{\mathcal B=(B_1,\ldots,B_k)}\ \Phi(\mathcal B).
}
\end{equation}
\end{theorem}

\begin{proof}
Let $Q$ denote the right-hand side of~\eqref{eq:minmax}. Let $t\le Q$. Consider any $k<n$ and pairwise disjoint bundles $B_1,\ldots,B_k$ with $v(B_r)<t$ for all $r\in[k]$. Since $\Phi(\mathcal B)\ge Q\ge t$ and $\max_{r}v(B_r)<t$, it follows that $\MMS_{n-k}(R(\mathcal B);v)\ge t$. Thus the remainder admits an $(n-k)$-partition where each bundle has value at least $t$, so $t$ is residual self-feasible and $\RMMS\ge Q$.

Conversely, by finiteness of $M$, the minimum defining $Q$ is achieved by some configuration $\mathcal B^*$. For every $t>Q$, each removed bundle in $\mathcal B^*$ has value at most $Q<t$, while $\MMS_{n-k}(R(\mathcal B^*);v)\le Q<t$. Thus $\mathcal B^*$ witnesses the failure of residual self-feasibility at threshold $t$, so $\RMMS\le Q$.
\end{proof}

\begin{corollary}\label{cor:upperbounds}
For every monotone valuation $v$, $\RMMS(M,v,n)\le\MMS(M,v,n)$. Moreover, $\RMMS(M,v,n)$ is at most the minimax share of the instance.
\end{corollary}

\section{The Exact Additive RMMS/MMS Frontier}\label{sec:exact}

Let $d=d_n$ be the largest odd integer at most $n$, and let $\lambda_n = 2d/(3d-1)$.

\subsection{The Density-Balance Lower Bound}

\begin{theorem}[Residual density balance~\cite{KurokawaEtAl2018,AkramiEtAl2026}]\label{thm:lower}
For every nonnegative additive valuation,
\begin{equation}\label{eq:lower}
\RMMS(M,v,n)\ge\lambda_n\MMS(M,v,n).
\end{equation}
\end{theorem}

\subsection{A Finite Bi-Valued Exact Witness}

Let $q=\lfloor(n-1)/2\rfloor$, $d=2q+1=d_n$, and $p=n-q$. Consider an instance with $2p-1$ high goods of value $d$ and $qd$ unit goods.

\begin{proposition}[Bi-valued exact witness]\label{prop:bivalue}
The bi-valued instance satisfies $\MMS=d+q$ and $\RMMS=d$, with support $\Theta(n^2)$.
\end{proposition}

\begin{proof}
If $n$ is odd, $d=n=2p-1$. Allocating one high good and $q$ unit goods per bundle gives $n$ bundles of value $d+q$, so $\MMS=d+q$. If $n$ is even, $d=n-1$ and there are $n+1$ high goods. Allocating two high goods to one bundle and one high good plus $q$ unit goods to each of the other $n-1=d$ bundles yields $\MMS=d+q$.

By Theorem~\ref{thm:lower}, $\RMMS\ge\lambda_n(d+q)=d$. Conversely, partition the $qd$ unit goods into $q$ bundles of value $d$. For any $t>d$, removing these $q$ bundles leaves $2p-1$ high goods for $p$ remaining agents. Each bundle of value $\ge t$ requires at least two high goods (requiring $2p$ goods), but only $2p-1$ exist. Hence $\RMMS\le d$.
\end{proof}

\subsection{A Linear-Support Sparse Exact Witness}

\begin{theorem}[Sparse exact witness]\label{thm:sparse}
For every $n\ge2$, consider an additive instance consisting of:
\begin{itemize}[nosep]
\item $2p-1$ high goods of value $d=2q+1$;
\item $2q$ goods of value $q$;
\item $q$ goods of value $1$.
\end{itemize}
Then $\MMS(M,v,n)=d+q$ and $\RMMS(M,v,n)=d$. The number of positive goods is linear:
\begin{equation}\label{eq:sparse-count}
|M|=
\begin{cases}
(5n-3)/2,&n\text{ odd},\\
(5n-4)/2,&n\text{ even}.
\end{cases}
\end{equation}
\end{theorem}

\begin{proof}
The low-valued block (total value $2q^2+q=qd$) admits two exact partitions:
\begin{enumerate}[label=(\alph*),nosep]
\item $d$ bundles of value $q$: $2q$ singletons of value $q$, plus one bundle of all $q$ unit goods;
\item $q$ bundles of value $d$: each containing two value-$q$ goods and one unit good ($q+q+1=d$).
\end{enumerate}

For $\MMS$, combining partition (a) with the high goods yields $n$ bundles of value at least $d+q$. For $\RMMS$, removing the $q$ bundles of value $d$ from partition (b) leaves $2p-1$ high goods for $p$ agents, which forces $\RMMS\le d$. Thus $\RMMS/\MMS = d/(d+q) = \lambda_n$.
\end{proof}

\begin{theorem}[Exact RMMS-to-MMS factor]\label{thm:exact}
For every $n\ge2$ and nonnegative additive valuations,
\begin{equation}\label{eq:exact}
\boxed{
\inf_{M,v:\MMS(M,v,n)>0} \frac{\RMMS(M,v,n)}{\MMS(M,v,n)} = \lambda_n = \frac{2d_n}{3d_n-1}.
}
\end{equation}
\end{theorem}

\begin{corollary}[Odd-even plateaus]\label{cor:plateaus}
For every $r\ge0$ with $n\ge2$, $\lambda_{2r+1}=\lambda_{2r+2}=\frac{2r+1}{3r+1}$, decreasing to $2/3$.
\end{corollary}

\begin{figure}[t]
\centering
\includegraphics[width=.78\textwidth]{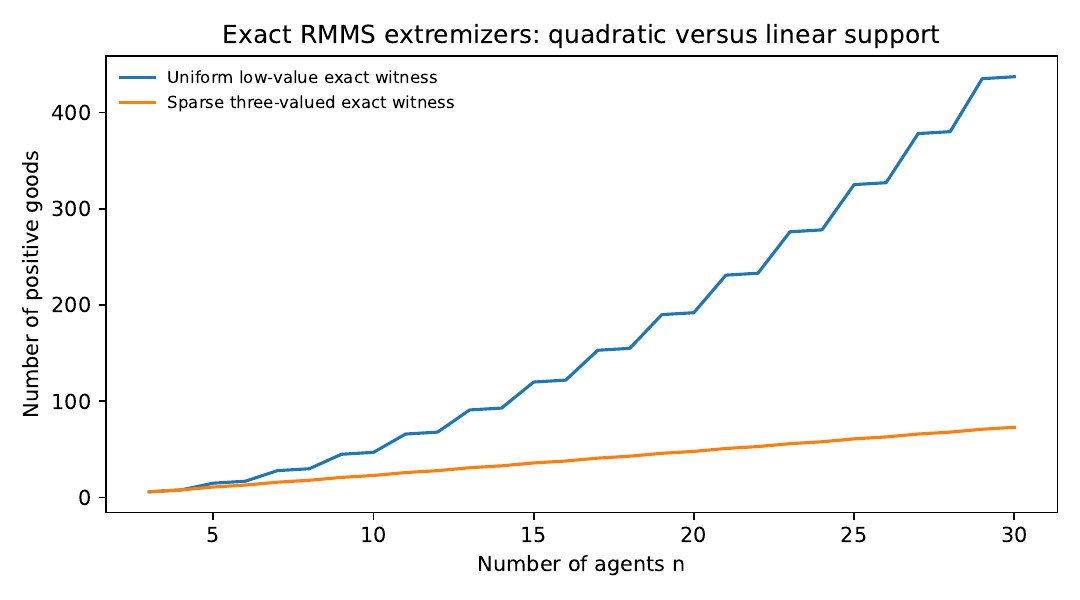}
\caption{Linear vs. quadratic support: our sparse construction achieves the exact frontier with $\Theta(n)$ items, compared to $\Theta(n^2)$ for uniform low-valued goods.}
\label{fig:support}
\end{figure}

\section{Minimum Support and Rigidity}\label{sec:support}

\begin{theorem}[Transportation-graph support bound]\label{thm:transport}
Let $q,d$ be positive coprime integers. Suppose a multiset $F$ of positive items admits:
\begin{enumerate}[nosep,label=(\roman*)]
\item a partition into $d$ bundles of value $q$;
\item a partition into $q$ bundles of value $d$.
\end{enumerate}
Then $|F|\ge q+d-1$. For $d=2q+1$, this bound equals $3q$ and is attained by Theorem~\ref{thm:sparse}.
\end{theorem}

\begin{proof}
Construct a bipartite graph with row vertices $R_1,\ldots,R_q$ (value-$d$ bundles) and column vertices $C_1,\ldots,C_d$ (value-$q$ bundles). Each good is an edge connecting the row and column containing it, weighted by its value.

In any connected component with $r$ rows and $c$ columns, total weight equality gives $rd=cq$. Since $\gcd(q,d)=1$, we have $q\mid r$ and $d\mid c$. Thus the graph is connected, requiring at least $q+d-1$ edges. For $d=2q+1$, $q+d-1=3q$, so the sparse block forms a spanning tree.
\end{proof}

\begin{theorem}[Rigidity at minimum support]\label{thm:rigidity}
Let $q\ge2$ and $d=2q+1$. Any multiset with minimum support $3q$ and at most two distinct positive values satisfying Theorem~\ref{thm:transport} is uniquely isomorphic to:
\[
2q\text{ goods of value }q
\qquad\text{and}\qquad
q\text{ goods of value }1.
\]
\end{theorem}

\begin{proof}
The transportation graph is a tree with $q+d$ vertices and $3q$ edges. Every edge has weight $\le q$ (column capacity). Thus each of the $q$ rows requires degree at least $3$. Since the sum of row degrees is $3q$, every row has degree exactly $3$.

A single positive value $x$ is impossible since $3x=d$ and $q/x=3q/(2q+1)\notin\mathbb Z$. Let the two values be $a>b>0$. Each row has degree $3$, so each row contains the same number of $a$-edges (either 1 or 2). If each row had one $a$-edge and two $b$-edges, $a+2b=2q+1$ with $a\le q$ implies $b\ge(q+1)/2>q/2$, which would force all columns to be singletons, yielding only $2q+1 < 3q$ edges.

Thus every row contains two $a$-edges and one $b$-edge, so $2a+b=2q+1$. If $a\le q/2$, then $b\ge q+1>q$, impossible. Hence $a>q/2$, which implies $a+b\ge q+1>q$. Therefore, no column containing an $a$-edge can contain any other positive edge. These columns must be singletons of value $q$, so $a=q$ and $b=1$.
\end{proof}

\begin{proposition}[Uniform goods require quadratic support]\label{prop:uniform}
If every good in $F$ has the same value $x$, then $x=1/k$ for an integer $k\ge1$ and $|F|\ge qd=\Theta(n^2)$.
\end{proposition}

\section{Optimality among Residual Self-Feasible Shares}\label{sec:domination}

\begin{theorem}[Optimal domination]\label{thm:general-domination}
Let $\mathcal C$ be any class of finite monotone valuations. The supremum factor $\rho$ such that a residual self-feasible share $s$ $\rho$-dominates $\MMS$ is:
\begin{equation}\label{eq:general-domination}
\inf_{v\in\mathcal C:\,\MMS>0} \frac{\RMMS(M,v,n)}{\MMS(M,v,n)}.
\end{equation}
The same supremum holds if $s$ is additionally required to be self-maximizing.
\end{theorem}

\begin{theorem}[Exact additive domination frontier]\label{thm:domination}
For additive valuations, $\RMMS$ $\lambda_n$-dominates $\MMS$, and no residual self-feasible share can $\rho$-dominate $\MMS$ for any $\rho>\lambda_n$.
\end{theorem}

\begin{center}
\small
\begin{tabular}{@{}lc@{}}
\toprule
Valuation Class & Best $\MMS$ Domination Factor among Residual Self-Feasible Shares\\
\midrule
Unit-demand & $1$\\
Additive & $\lambda_n=2d_n/(3d_n-1)$\\
Submodular & $1/n$\\
Subadditive & $1/n$\\
\bottomrule
\end{tabular}
\end{center}

\begin{corollary}[Lone-divider barrier]\label{cor:barrier}
On additive valuations, no share-based lone-divider algorithm certified via residual self-feasibility can guarantee a uniform $\MMS$ factor greater than $\lambda_n$.
\end{corollary}

\begin{corollary}[Simultaneous fairness guarantees]\label{cor:simultaneous}
For additive valuations, there exists:
\begin{enumerate}[label=(\roman*),nosep]
\item a partial allocation that is simultaneously $\EFX$ and $\lambda_n$-$\MMS$;
\item a complete allocation that is simultaneously $\EFL$ (hence $\mathrm{EF1}$) and $\lambda_n$-$\MMS$.
\end{enumerate}
\end{corollary}

\section{Threshold Packing-Covering Cuts}\label{sec:cut}

Let $v$ be an integer additive valuation and $T\in\mathbb N_{>0}$. For $U\subseteq M$, define $\bp_T(U)$ as the minimum number of bundles partitioning $U$ with each bundle of value strictly below $T$ (bin packing with capacity $T-1$). For $R\subseteq M$, define $\cov_T(R)$ as the maximum number of disjoint bundles in $R$ of value at least $T$.

\begin{theorem}[Threshold packing-covering cut]\label{thm:cut}
For every integer additive instance and integer $T>0$,
\begin{equation}\label{eq:cut}
\boxed{
T\le\RMMS(M,v,n) \quad\Longleftrightarrow\quad \min_{U\subseteq M}\bigl(\bp_T(U)+\cov_T(M\setminus U)\bigr)\ge n.
}
\end{equation}
\end{theorem}

\begin{proof}
If the minimum is at most $n-1$, let $U$ satisfy $b=\bp_T(U)$ and $c=\cov_T(M\setminus U)$ with $b+c\le n-1$. Removing the $b$ sub-$T$ bundles from $U$ leaves at most $c\le n-b-1$ threshold bundles in $M\setminus U$, so $T$ is not residual self-feasible.

Conversely, if $T$ is not residual self-feasible, there exist $k<n$ removed bundles of value $<T$ covering some set $U$, such that $M\setminus U$ cannot form $n-k$ bundles of value $\ge T$. Thus $\bp_T(U)\le k$ and $\cov_T(M\setminus U)\le n-k-1$, giving sum $\le n-1$.
\end{proof}

\section{Related Work}\label{sec:prior}

$\MMS$ was introduced by Budish~\cite{Budish2011} and analyzed extensively~\cite{KurokawaEtAl2018,AmanatidisSurvey2023,BabaioffFeige2025}. Feige~\cite{Feige2025} and Akrami et al.~\cite{AkramiEtAl2026} established $\RMMS$ and its connection to lone-divider algorithms and simultaneous fairness (EFX/EFL + MMS)~\cite{AkramiRathi2025}. Our work establishes the exact finite-agent frontier $\lambda_n$, proves support minimality and rigidity, and provides structural cut characterizations.

\section{Conclusion}

In this paper, we have determined the exact finite-agent additive frontier of the residual maximin share ($\RMMS$) relative to the classical maximin share ($\MMS$), proving that
\[
\inf_{M,v:\MMS>0} \frac{\RMMS(M,v,n)}{\MMS(M,v,n)} = \lambda_n = \frac{2d_n}{3d_n-1},
\]
where $d_n$ is the largest odd integer at most $n$. The exact frontier forms consecutive odd-even plateaus ($\lambda_{2r+1}=\lambda_{2r+2}=\frac{2r+1}{3r+1}$) and converges monotonically to $2/3$. This establishes that the lower bounds previously derived from residual density-balance analyses are exact for every finite $n\ge2$.

Beyond the numerical frontier, our analysis establishes several structural properties of dynamic fair division:
\begin{enumerate}[leftmargin=2em]
\item \textbf{Linear Support and Tree Minimality:} While standard bi-valued constructions require $\Theta(n^2)$ items, we constructed an explicit three-valued family achieving the exact boundary with linear support ($(5n-3)/2$ for odd $n$ and $(5n-4)/2$ for even $n$). By formulating the dual partitions as a bipartite transportation graph, we proved that $3q$ items is the absolute minimum support for any valid filler, forcing the graph to be a spanning tree. Under two values, the $(q,q,1)$ configuration is uniquely rigid.
\item \textbf{Theoretical Limits of the Lone-Divider Invariant:} Because $\RMMS$ is pointwise maximal and self-maximizing among all residual self-feasible shares, $\lambda_n$ constitutes the exact upper bound for any recursive fair division algorithm certified via residual self-feasibility.
\item \textbf{Structural Min--Max and Threshold Cuts:} We provided a finite min--max representation for general monotone valuations and an exact threshold packing-covering cut characterization ($\bp_T(U)+\cov_T(M\setminus U)\ge n$) for integer additive valuations.
\end{enumerate}

These structural characterizations suggest several directions for future work. A natural algorithmic problem is whether the threshold packing-covering characterization can be exploited to design pseudo-polynomial or parameterized approximation algorithms for computing $\RMMS$. Another open direction is studying whether alternative dynamic invariants can be designed for sequential allocation mechanisms that bypass the share-based lone-divider barrier.

\appendix

\section{Appendix A: Algebraic Identities for the Exact Ratio}

For odd $n$, $d=n$ and $q=(n-1)/2$, so
\[
\frac{d}{d+q} = \frac{n}{n+(n-1)/2} = \frac{2n}{3n-1}.
\]
For even $n$, $d=n-1$ and $q=(n-2)/2$, so
\[
\frac{d}{d+q} = \frac{n-1}{n-1+(n-2)/2} = \frac{2n-2}{3n-4}.
\]
Writing $n=2r+1$ or $n=2r+2$ yields the common plateau $(2r+1)/(3r+1)$.

\section{Appendix B: First Sparse Exact Witnesses}

Table~\ref{tab:sparse-witnesses} lists the explicit parameters of the sparse exact witness for small values of $n$.

\begin{table}[h]
\centering
\small
\begin{tabular}{@{}c c c c c c c@{}}
\toprule
$n$ & $q$ & $d$ & High Goods & Value-$q$ Goods & Unit Goods & $(\MMS,\RMMS)$\\
\midrule
2 & 0 & 1 & $3\times1$ & -- & -- & $(1,1)$\\
3 & 1 & 3 & $3\times3$ & $2\times1$ & $1\times1$ & $(4,3)$\\
4 & 1 & 3 & $5\times3$ & $2\times1$ & $1\times1$ & $(4,3)$\\
5 & 2 & 5 & $5\times5$ & $4\times2$ & $2\times1$ & $(7,5)$\\
6 & 2 & 5 & $7\times5$ & $4\times2$ & $2\times1$ & $(7,5)$\\
7 & 3 & 7 & $7\times7$ & $6\times3$ & $3\times1$ & $(10,7)$\\
8 & 3 & 7 & $9\times7$ & $6\times3$ & $3\times1$ & $(10,7)$\\
\bottomrule
\end{tabular}
\caption{Parameters of the sparse exact witness for $n\le 8$.}
\label{tab:sparse-witnesses}
\end{table}

\end{document}